\documentclass[11pt,book,reqno]{amsart}
\usepackage[T1]{fontenc}
\usepackage[latin9]{inputenc}
\usepackage{amssymb}
\usepackage{color}
\usepackage{graphicx,color}
\usepackage{amsmath, amssymb, graphics}
\usepackage{qcircuit}
\usepackage{hyperref}

\makeatletter
\numberwithin{equation}{section} %% Comment out for sequentially-numbered
\numberwithin{figure}{section} %% Comment out for sequentially-numbered
  \@ifundefined{theoremstyle}{\usepackage{amsthm}}{}
  \theoremstyle{plain}
  \newtheorem{thm}{Theorem}[section]
  \theoremstyle{plain}
  
  \theoremstyle{plain}
  
  \theoremstyle{remark}
  \newtheorem{rem}[thm]{Remark}
  \theoremstyle{remark}
  
  \theoremstyle{plain}
  
  \newtheorem{mydef}{Definition}

\smallskip
\def\<{{\langle }}
\def\>{{\rangle }}

\def\ket#1{|#1\rangle}

\smallskip
\def\<{{\langle }}
\def\>{{\rangle }}

\begin{document}

\title[Schmidt representation Real 3-qubits]{Schmidt representation of 3-qubits with real amplitudes}

\author{Oscar Perdomo}
\date{\today}
\begin{abstract} 
From the Schmidt representation we have that, up to local gates, every 2-qubit  state can be written as $\ket{\phi}=\lambda_1 \ket{00}+\lambda_2 \ket{11}$ with $\lambda_1$ and $\lambda_2$ real numbers. For 3-qubits states, it is known, \cite{AA}, that up to local gates every 3-qubit  state can be written as  $\ket{\phi}=\lambda_1 \ket{000}+\lambda_2 e^{i \theta }\ket{100}+\lambda_3\ket{101}+\lambda_4\ket{110}+\lambda_5\ket{111}$ with $\lambda_i\ge0$ and $0\le \theta \le \pi$. In this paper, we show that no every 3-qubit state with real amplitudes can be transform in the form  $\ket{\phi}=\lambda_1 \ket{000}+\lambda_2 \ket{100}+\lambda_3\ket{101}+\lambda_4\ket{110}+\lambda_5\ket{111}$ by using local gates in the orthogonal group (the group generated by $R_y(\theta)$ and $X$ gates). We also show that, up to local gates in the orthogonal group, every 3-qubit with real amplitudes can be written as $\ket{\phi}=\lambda_1 \ket{000}+\lambda_2 \ket{011}+\lambda_3\ket{101}+\lambda_4\ket{110}+\lambda_5\ket{111}$ with the $\lambda_i$ real numbers. An explanation of this result can be found in the youtube video \url{https://youtu.be/gDN20QHzsoQ}

\end{abstract}
%\subjclass[2000]{53C42, 53A10}
\maketitle

\section{introduction}

Two pure $n$-qubit states $\ket{\phi_1}$ and $\ket{\phi_2}$ define the same entanglement if there is a local gate in the unitary group, a matrix $U=U_1\otimes\dots\otimes U_n$ with each $U_i$ a 2 by 2 unitary matrix, such that $\ket{\phi_2}=U\ket{\phi_1}$. This relation between $n$-qubits is an equivalent relation and the quotient space is denoted as $\mathcal{E}_n$ and it is called the space of entanglements types, \cite{W}. In this space, we think of all the $n$-qubit states that are connected with local gates as a single element. We have another interpretation for this partition of the set of $n$-qubit states:  Two $n$ qubit states $\ket{\phi_1}$ and $\ket{\phi_2}$ define the same entanglement if and only if they are LOCC equivalent, this is, if not only $\ket{\phi_2}$  transform into $\ket{\phi_1}$ using local operation and classical communications (LOCC) but also $\ket{\phi_1}$  transform into $\ket{\phi_2}$ using LOCC, \cite{B}.

Even though entanglement is an important aspect of Quantum information and Computation \cite{N}, we could say that it is only completely understood for 2-qubit states. For example, it is not known if any pair of 3-qubit states can be connected with a circuit  made of local gates and at most three $CZ$ gates. This is, we do not know if we can move from one entanglement type to another using three $CZ$ gates of less. It is known that we can do this with four $CZ$ gates. On the other hand it is known that every 3-qubit state can be reached from the state $\ket{000}$ by means of local gates and at most 3 $CZ$ gates, \cite{Z}.

There is not doubt that to deal with real numbers and orthogonal matrices is easier (at least, computationally less expensive) than to deal with complex numbers and unitary matrices. We need only one real parameter to describe an orthogonal 2 by 2 matrix, while we need 4 real  parameters (three if we take unitary matrices with determinant one) to describe  a 2 by 2 unitary matrix. Let us say that 

$$\ket{\phi}=a_0 \ket{0\dots 00}+a_1\ket{0\dots 01}+\dots+a_{2^n-1} \ket{1\dots 1} $$

is a {\it real} $n$ qubit state if all the $a_i$ are real numbers. These states have been considered earlier as part of Real-Vector space Quantum Theory \cite{HW}. We say that two real $n$-qubit states $\ket{\phi_1}$ and $\ket{\phi_2}$ define the same real entanglement if there is a local gate in the orthogonal group, a matrix $U=U_1\otimes\dots\otimes U_n$ with each $U_i$ a 2 by 2 orthogonal matrix, such that $\ket{\phi_2}=U\ket{\phi_1}$. This relation between real $n$-qubits is an equivalent relation and the quotient space is denoted as $\mathcal{E}^{\mathbb{R}}_n$. There is a natural map $\chi_n:\mathcal{E}^{\mathbb{R}}_n\longrightarrow \mathcal{E}_n$ that sends the real entanglement of the real $n$-qubit $\ket{\phi}$ into the entanglement of $\ket{\phi}$ viewed as a regular pure state with complex amplitudes. For the case of 2-qubits, the map $\chi_2$ is one to one and onto, this is, we do not miss any entanglement by considering only real 2-qubit states and by considering only orthogonal gates, \cite{P}. This can be useful, for example, if we are training a quantum circuit (\cite{B1}, \cite{L}, \cite{PG})  to achieve a particular entanglement, there is not need to consider unitary matrices, it is enough to consider real 2-qubits states and $R_y$ gates.  For 3-qubits, the map $\chi_3:\mathcal{E}^{\mathbb{R}}_3\longrightarrow \mathcal{E}_3$ is not onto because we need five real parameters to describe  $\mathcal{E}_3$ (\cite{SM}, \cite{AA}, \cite{W}), while we only need four to describe $ \mathcal{E}^{\mathbb{R}}_3$. Moreover, $\chi_3$ is not one to one, since we can show that the GHZ state $\frac{1}{\sqrt{2}}(\ket{000}+\ket{111})$ and the state $\frac{1}{2}(\ket{001}-\ket{010}+\ket{100}+\ket{111})$ have different real entanglement (there is not local orthogonal gates connecting them) and they have the same entanglement because they can be connected with a unitary local gate, \cite{P1}.

One of the goals of this paper is to show that for any real state $\ket{\phi_1}$ there is a real state $\ket{\phi_2}=R_y(\theta_2)\otimes R_y(\theta_1)\otimes R_y(\theta_0)\ket{\phi_1}$ of the form $\lambda_1 \ket{000} +\lambda_2 \ket{011}+\lambda_3\ket{101}+\lambda_4\ket{110}+\lambda_5\ket{111}$ with $\sum\lambda_i^2=1$. Recall that 

$$R_y(\theta)=\left(
\begin{array}{cc}
 \cos \left(\frac{\theta}{2}\right) & -\sin \left(\frac{\theta}{2}\right) \\
 \sin \left(\frac{\theta}{2}\right) & \cos \left(\frac{\theta}{2}\right) \\
\end{array}
\right)\, .$$ 

We would like to point out that if we try to follow the lines presented in \cite{AA} to try to eliminate the coefficients of $\ket{001},\ket{010},\ket{011}$ we will fail due to the fact that the real numbers are not algebraically complete and not every quadratic equation has a solution in the real numbers.  This is not only a proof technicality, since we can show that it is impossible to use local orthogonal gates to transform the state $\xi=\frac{1}{\sqrt{2}}\left(\ket{000}+\ket{001}+\ket{011}+\ket{101}-\ket{110}\right)$ into a state of the form $\lambda_1 \ket{000}+\lambda_2\ket{100}+\lambda_3\ket{101}+\lambda_4\ket{110}+\lambda_5\ket{111}$ with $\lambda_i$ real numbers. We have that the Schmidt representation (the one presented in \cite{AA} ) of the state $\xi$  is 

$$\frac{1}{\sqrt{10}} \left(2\ket{000}-i\ket{100}+\ket{101}+\ket{110}+\sqrt{3} \ket{111} \right)$$

Therefore, in order to find a Schmidt representation for real 3-qubits, the first task to eliminate three coefficients among real 3-qubit states is to decide which coefficients can be transform into zero. It turns out that we can easily make the coefficients of $\ket{001}$ and $\ket{100}$ vanish, see Theorem \ref{thm1}. Therefore we have a representation of every real 3-qubit state in the 5 dimensional sphere $S_0^5 $ of qubits $\sum u_{rst} \ket{rst}$ with $u_{010}=u_{100}=0$. In order to prove that we can make an additional coefficient equal to zero we need to learn how to navigate this 5 dimensional sphere. We define a vector field $X$ in $S_0^5 $ that shows us the right direction to  move in $S_0^5 $ in such a way that we only touch states that are connected with local orthogonal gates. In a little more precise words, we can visualize the integral curves of the vector field $X$ as bridges form by states that define the same real entanglement, see Theorem \ref{X}. We manage to show that we can make the coefficients $u_{001}=u_{010}=u_{100}=0$ by studying these bridges, this is, by studying the integral curves of vector field $X$. The author would like to thank Jos\'e Ignacio Latorre for sharing his knowledge.

%%End Introduction
%% Beginning section "Representation in a 5 dimensional space"

\section{Representation in a 5 dimensional space}

In this section we show that, up to local gates in the orthogonal group, we can write every real 3-qubit state as
\begin{eqnarray}\label{sr1}
\ket{\phi}=x_1 \ket{000}+x_2 \ket{010} +x_3 \ket{011}+x_4\ket{101}+x_5\ket{110}+x_6\ket{111}
\end{eqnarray}
More precisely, we have,

\begin{thm}\label{thm1} Let $\ket{\phi}=u_0 \ket{000}+u_1 \ket{001}+u_2 \ket{010}+u_3 \ket{011}+u_4 \ket{100}+u_5 \ket{101}+u_6 \ket{110}+u_7 \ket{111}$, with $u_i$ real numbers. There exist $\theta_0$  and $\theta_2$ such that 

$$R_y(\theta_2)\otimes R_y(0)\otimes R_y(\theta_0) \ket{\phi}=x_1 \ket{000}+x_2 \ket{010} +x_3 \ket{011}+x_4\ket{101}+x_5\ket{110}+x_6\ket{111}$$

\end{thm}
\begin{proof} A direct computation shows that if $\theta_0$ satisfies the equation $a_1 \sin (\theta_0)+a_2 \cos (\theta_0)=0$ with $a_1=\left(-u_0^2+u_1^2-u_4^2+u_5^2\right)$ and $a_2=-2 (u_0 u_1+u_4 u_5)$, and  $\theta_2$ satisfies the equation $b_1 \sin (\frac{\theta_2}{2})+b_2 \cos (\frac{\theta_2}{2})=0$ with 

$$b_1= -u_4 \sin \left(\frac{\theta_0}{2}\right)-u_5 \cos \left(\frac{\theta_0}{2}\right)\quad\hbox{and}\quad b_2= u_1 \cos \left(\frac{\theta_0}{2}\right) +u_0 \sin \left(\frac{\theta_0}{2}\right) $$

then, $R_y(\theta_2)\otimes R_y(0)\otimes R_y(\theta_0) \ket{\phi}$ can be written as $x_1 \ket{000}+x_2 \ket{010} +x_3 \ket{011}+x_4\ket{101}+x_5\ket{110}+x_6\ket{111}$ for some real numbers $x_1\dots x_6$.

\end{proof}

%Recall that the equation for $\theta_0$ is either true for all $\theta_0$ or can be easily solved in terms of either the function $\tan^{-1}$ or $\cot^{-1}$.

The theorem above tell us that every  possible entanglement create by a real 3-qubit states is realized by a state in the 5-dimensional sphere $S_0^5$ defines as follow

\begin{mydef} We define 

$$S_0^5= \{x_1 \ket{000}+x_2 \ket{010} +x_3 \ket{011}+x_4\ket{101}+x_5\ket{110}+x_6\ket{111}: \,  \, x_1^2+\dots + x_6^2=1\,\}$$

\end{mydef}

Counting dimensions we have that it is expected that, in general, for any 3-qubit state in $\ket{\phi_0}\in S_0^5$, there is a curve of states in $S_0^5$ that can be reached by $\ket{\phi_0}$ using orthogonal local gates. To see the previous observation we notice that  all the real states are described by points in the $7$-dimensional sphere of unit vectors in $\mathbb{R}^8$. Since the lie group $O(2)$ of orthogonal 2 by 2 matrices is 1 dimensional, then, the space of local orthogonal gates $O(2)\otimes O(2)\otimes O(2)$ is a three dimensional manifold. With this in mind it is expected that, up to local orthogonal gates, the space of 3-qubit states with real amplitudes is described with $7-3=4$ parameters.

\begin{rem}\label{mm} If we identify  $\ket{\phi}=u_0 \ket{000}+u_1 \ket{001}+u_2 \ket{010}+u_3 \ket{011}+u_4 \ket{100}+u_5 \ket{101}+u_6 \ket{110}+u_7 \ket{111}$ with the unit vector in $(u_0,\dots, u_7)\in \mathbb{R}^8$, then  $S_0^5$ is identified with the manifold 
\begin{eqnarray}
M=\{ (x_1,0,x_2,x_3,0,x_4,x_5,x_6)\in \mathbb{R}^8\,:\, x_1^2+\dots +x_6^2=1\,\}
\end{eqnarray}
\end{rem}

The following tangent vector field on $S_0^5$ has the property that each one of its integral curves consists of states connected by local orthogonal gates.

\begin{thm} \label{X}  In this theorem we are using the identification given in Remark \ref{mm}.
The vector field $X=(X_1,0,X_2,X_3,0,X_4,X_5,X_6)$ with 
\begin{eqnarray*}
X_1&=& x_2 \left(x_1^2-x_4^2\right),\quad X_2=-x_1^3+\left(x_3^2+x_4^2+x_5^2\right) x_1+2 x_3 x_4 x_5\\
X_3 &=& (x_3 x_4+x_1 x_5) x_6-x_2 (x_1 x_3+x_4 x_5),\quad X_4=\left(x_1^2-x_4^2\right) x_6,\\ X_5&=& (x_1 x_3+x_4 x_5) x_6-x_2 (x_3 x_4+x_1 x_5), \, X_6=2 x_4^3+\left(x_2^2+x_6^2-1\right) x_4-2 x_1 x_3 x_5
\end{eqnarray*}

defines a tangent vector field in $M$. Moreover, every solution 

$$\ket{\phi(t)}=x_1(t) \ket{000}+x_2(t) \ket{010} +x_3(t) \ket{011}(t)+x_4(t)\ket{101}+x_5(t)\ket{110}+x_6(t)\ket{111}$$

 of the initial value problem 

$$\ket{\phi(t)}^\prime=X(\ket{\phi(t)}) \hbox{, with } \ket{\phi}(0)=\ket{\phi_0}$$

 satisfies that any pair of 3-qubit states $\ket{\phi(t_1)} $ and  $\ket{\phi(t_2)} $ can be transform into each other using a local orthogonal gate. Additionally, we have that 
 
 $$\ket{\phi(t)}=R_y(\theta_2(t))\otimes R_y(\theta_1(t))\otimes R_y(\theta_0(t)) \ket{\phi}$$
 
 with 
 
 \begin{eqnarray*}
\theta_0(t)&=& \int_0^tL_0\left(\ket{\phi}(\tau)\right)d\tau \hbox{ where } L_0=-2 (x_1 x_3+x_4 x_5)\\
\theta_1(t)&=& \int_0^tL_1\left(\ket{\phi}(\tau)\right)d\tau \hbox{ where } L_1=2 (x_4^2-x_1^2)\\
\theta_2(t)&=& \int_0^tL_2\left(\ket{\phi}(\tau)\right)d\tau \hbox{ where } L_2=-2 (x_3 x_4+x_1 x_5)
\end{eqnarray*}
\end{thm}

\begin{proof} The fact that $X$ is tangent to  $M$ follows by verifying that the vector $X$ is perpendicular to the vector $(x_1,0,x_2,x_3,0,x_4,x_5,x_6)$. Recall that a vector $v\in\mathbb{R}^8$ is tangent to the sphere $M$ at the point $p$ if and only the vector $v$ has  its second and fifth entry equal to zero and it is perpendicular to the vector $p$. To proof the rest of the theorem, for any $\ket{\phi}\in M$, we define the manifold

$$\Sigma_{\ket{\phi}}=\{\psi(\theta_2,\theta_1,\theta_0)=R_y(\theta_2)\otimes R_y(\theta_1)\otimes R_y(\theta_0) \ket{\phi}: \theta_2,\theta_1,\theta_0\in\mathbb{R}\}\subset \mathbb{R}^8$$

A direct verification shows that for any $\ket{\xi}\in \Sigma_{\ket{\phi}}$, the vector field  $X$ is a linear combination of the vectors $\frac{\partial \psi}{\partial \theta_2}(\ket{\xi})$, $\frac{\partial \psi}{\partial \theta_1}(\ket{\xi})$ and $\frac{\partial \psi}{\partial \theta_0}(\ket{\xi})$. Therefore the integral curve of the vector field $X$ containing the state $\ket{\phi}$ is part of the intersection of the manifolds $\Sigma_{\ket{\phi} }$ and $M$. Hence, any pair of 3-qubit states $\ket{\phi(t_1)} $ and $\ket{\phi(t_2)} $ in an integral curve of $X$ are connected by a local orthogonal gate. The formula for the angles follows from the fact that $X=L_0\frac{\partial \psi}{\partial \theta_0}+L_1\frac{\partial \psi}{\partial \theta_1}+L_2\frac{\partial \psi}{\partial \theta_2}$.
\end{proof}

The following functions are first integrals for the vector field $X$. 

\begin{thm} The vector field $X$ defined on $M$ has the following first integrals

\begin{eqnarray*} I_2&=&2 x_3^4+2 \left(x_2^2+2 x_4^2+2 x_6^2-1\right) x_3^2+4 x_2 x_5 x_6 x_3+2 x_4^4+2 x_6^4+2 x_5^2 x_6^2-\\
& &2 x_6^2+x_4^2 \left(4 x_6^2-2\right)+1\\
I_3&=&x_1^4+2 \left(x_2^2+x_4^2\right) x_1^2+4 x_2 x_4 x_6 x_1+x_2^4+x_3^4+x_4^4+x_5^4+x_6^4+2 x_3^2 x_5^2+\\
& &2 x_3^2 x_6^2+2 x_4^2 x_6^2+2 x_5^2 x_6^2+2 x_2^2 \left(x_3^2+x_5^2+x_6^2\right)\\
I_4&=&2 x_4^4+\left(4 x_5^2+4 x_6^2-2\right) x_4^2+2 x_5^4+2 x_6^4+2 x_3^2 x_6^2-2 x_6^2+4 x_2 x_3 x_5 x_6+\\
& &2 x_5^2 \left(x_2^2+2 x_6^2-1\right)+1\\
I_0&=&(x_2 x_4-x_1 x_6)^2+4 x_1 x_3 x_4 x_5 
\end{eqnarray*}

\end{thm}

We decided to keep the notation used in \cite{S} for the invariants $I_2$, $I_3$ and $I_4$ where it is explained how these first  integrals correspond to  $tr(\rho^2_C)$, $tr(\rho^2_B)$ and $tr(\rho^2_A)$ respectively. 

We finish this section by showing the equilibrium points of the vector field $X$.

\begin{thm} \label{equili} The equilibrium point of the vector field $X$ is the union of three 3-dimensional spheres, and a 1-dimensional manifold $P$. More precisely, they are $S_1^3\cup S_2^2\cup S_3^3\cup P$ where
\begin{eqnarray*} S_1^3&=&\{x\in M: x_1=x_4, x_3=-x_5\} \\
  S_2^3&=&\{x\in M: x_1=-x_4, x_3=x_5\} \\
S_3^3&=&\{x\in M: x_1=x_4=0\}\\
 P&=&\{x\in M: -2 x_1^3+x_1+2 x_3 x_4 x_5=2 x_4^3-x_4-2 x_1 x_3 x_5=x_2=x_6=0 \}
\end{eqnarray*}
\end{thm}

\section{Representation in a 4-dimensional space}

In this section we show that up to a local orthogomal gate, every real 3-qubit state can be written in the form
$\ket{\phi(t)}=x_1 \ket{000}+x_3 \ket{011}+x_4\ket{101}+x_5\ket{110}+x_6\ket{111}$. To prove this result we will be using the vector field $X$ defined in the previous section. We will first show that the result is true for the equilibrium points and then we will show the case for all other points in the space $M$ defined in Remark \ref{mm}.

\begin{thm} \label{thm2} Let $\ket{\phi}=u_0 \ket{000}+u_1 \ket{001}+u_2 \ket{010}+u_3 \ket{011}+u_4 \ket{100}+u_5 \ket{101}+u_6 \ket{110}+u_7 \ket{111}$, with $u_i$ real numbers. There exist $\theta_0$, $\theta_1$   and $\theta_2$ such that 

$$R_y(\theta_2)\otimes R_y(\theta_1 )\otimes R_y(\theta_0) \ket{\phi}=x_1 \ket{000} +x_3 \ket{011}+x_4\ket{101}+x_5\ket{110}+x_6\ket{111}$$

\end{thm}

\begin{proof}
By Theorem \ref{thm1} we can assume that $\ket{\phi}=w_1 \ket{000}+w_2 \ket{010} +w_3 \ket{011}+w_4\ket{101}+w_5\ket{110}+w_6\ket{111}$. Let us consider first the case that $\ket{\phi}$ is an equilibrium point of the vector field $X$. By theorem \ref{equili} we have that $\ket{\phi}$ is in either $S_1^3$, $S_2^3$, $S_3^3$, or $P$. If $\ket{\phi}\in P$ then $w_2=0$ and the theorem is trivially  true. Let us now assume that $\ket{\phi}\in S_3^3$. Then we have that $w_1=w_4=0$ and a direct computation shows that if we pick $\theta_0=\theta_1=0$ and $\theta_2$ an angle that satisfies $w_2 \cos \left(\frac{\theta_2}{2}\right)-w_5 \sin \left(\frac{\theta_2}{2}\right)=0$, then we have that $R_y(\theta_2)\otimes R_y(\theta_1 )\otimes R_y(\theta_0) \ket{\phi}$ has the desired form. Let us assume now that $\ket{\phi}$ is in $S_1^3$. Then we have that $w_1=w_4$ and $w_5=-w_3$. 
A direct computation shows that 

\begin{eqnarray*}
R_y(\theta_2)\otimes R_y(\theta_1 )\otimes R_y(-\theta_2) \ket{\phi}&=&x_1 \ket{000}+z_1 \ket{001}+x_2 \ket{010} +x_3 \ket{011}-z_1 \ket{100}\\
& &+x_4\ket{101}+x_5\ket{110}+x_6\ket{111}
\end{eqnarray*}

with 

\begin{eqnarray}\label{eq1}
z_1=\frac{1}{2} (w_2 \sin \theta_2+w_6 \sin\theta_2-2 w_3 \cos \theta_2)\sin \frac{\theta_1}{2} -w_1 \sin \theta_2\cos \frac{\theta_1}{2}
\end{eqnarray}

and
\begin{eqnarray}\label{eq2}
x_2=w_1 \sin\frac{\theta_1}{2} \cos \theta_2+\frac{1}{2} \cos \frac{\theta_1}{2} (2 w_3 \sin \theta_2+w_2 \cos \theta_2+w_6 \cos \theta_2+w_2-w_6)
\end{eqnarray}

We need to find $\theta_1$ and $\theta_2$ such that $x_2$ and $z_1$ are equal to zero. Let us make

\begin{eqnarray}\label{rem}
\sin \frac{\theta_1}{2}=\lambda w_1 \sin \theta_2\hbox{ and } \cos \frac{\theta_1}{2}=\frac{\lambda}{2} (w_2 \sin \theta_2+w_6 \sin\theta_2-2 w_3 \cos \theta_2)
\end{eqnarray}

With this substitution $z_1=0$ and $x_2=\lambda f$ with 

$$f=w_1^2 \sin \theta_2 \cos \theta_2+\frac{1}{4} (2 w_3 \sin \theta_2+(w_2+w_6) \cos \theta_2+w_2-w_6) ((w_2+w_6) \sin \theta_2-2 w_3 \cos \theta_2)$$

A direct computation shows that 

$$f(0)+f(\frac{\pi }{2})+f(\pi)+f(\frac{3\pi }{2})=0$$

Therefore by the intermediate value theorem we have that for some angle $\theta_2$ between $0$ and $\frac{3\pi}{2}$ we have that $f(\theta_2)=0$. Once we have $\theta_2$ is easy to find $\theta_1$ and $\lambda$ that satisfy Equation \ref{rem}. Therefore the result is also true for points in $S_1^3$. For points in $S_2^3$ the proof is similar but in this case we use the fact that since $w_4=-w_1$ and $w_5=w_3$, then a direct computation shows that

\begin{eqnarray*}
R_y(\theta_2)\otimes R_y(\theta_1 )\otimes R_y(\theta_2) \ket{\phi}&=&x_1 \ket{000}+z_1 \ket{001}+x_2 \ket{010} +x_3 \ket{011}+z_1 \ket{100}\\
& &+x_4\ket{101}+x_5\ket{110}+x_6\ket{111}
\end{eqnarray*}

Finally, let us consider the case where $\ket{\phi}$ is not an equilibrium point. Assume

%$$\ket{\xi(t)}=x^0_1(t) \ket{000}+x^0_2(t) \ket{010} +x^0_3(t) \ket{011}(t)+x^0_4(t)\ket{101}+x^0_5(t)\ket{110}+x^0_6(t)\ket{111}$$
$$\ket{\xi(t)}=x_1(t) \ket{000}+x_2(t) \ket{010} +x_3(t) \ket{011}(t)+x_4(t)\ket{101}+x_5(t)\ket{110}+x_6(t)\ket{111}$$

is the integral curve of the vector field $X$ satisfying $\ket{\xi(0)}=\ket{\phi}$. We will be using the fact that every integral curve of the vector field $X$ is either topologically a circle or it  contains a sequence of points that converges to an equilibrium point. Let us consider the case that $w_1=w_4$. Since 
\begin{eqnarray}\label{dex1x4}
x_1^\prime(t)=x_2(t) (x_1^2(t)-x_4^2(t))\hbox{ and } x_4^\prime(t)=x_6(t) (x_1^2(t)-x_4^2(t))
\end{eqnarray}

Then, by the existence and uniqueness theorem of differential equations we have that $x_1(t)=w_1$ and $x_4(t)=w_4$ for all $t$. Therefore 

$$x_2^\prime(t)=w_1 (x_3(t)+x_5(t))^2$$

Since we are assuming that $\ket{\phi}$ is not an equilibrium point then $w_1=w_4\ne0$ and $w_3\ne -w_5$. Therefore $x_2(t)$ is a non-constant function that satisfies $x_2^\prime(t)\ge 0$ or $x_2^\prime(t)\le 0$. In particular $x_2(t)$ cannot be a periodic solution and, as a consequence,  the solution $\ket{\phi(t)}$ is not periodic. Therefore there exists  a sequence of points $t_1,t_2,\dots$ such that $\ket{\xi(t_i)}$ converges to an equilibrium point $\ket{\phi_0}$. By continuity we have that $\ket{\phi}$ differ from $\ket{\phi_0}$ by a local orthogonal gate. Since we already have shown that the theorem holds true for the equilibrium points of $X$, then the result follows in this case. The proof is similar if we assume that $w_4=-w_1$. In the case that neither $w_4=w_1$ nor $w_4(t)=-w_1$, then, since $x_1(t)$ and $x_4(t)$ satisfies the system \ref{dex1x4}, then we have that $x^2_1(t)\ne x^2_4(t)$ for all $t$. Recall that if $x_1(t)$ is not a periodic function then the result follows due to the fact that the integral curve will be arbitrarily close to an equilibrium point. Under the assumption that $x_1(t)$ is a periodic function we have that $x_1^\prime(t_0)$ must be zero for some $t_0$. Since $ x_1^\prime(t_0)=x_2(t_0) (x_1^2(t_0)-x_4^2(t_0))$ then we must have that $x_2(t_0)=0$ and therefore the theorem follows also in this case.

\end{proof}

\end{document}